\documentclass[12pt,oneside]{amsart}
\usepackage[latin9]{inputenc}
\usepackage[a4paper]{geometry}
\usepackage{enumitem}
\usepackage{amstext}
\usepackage{amsthm}
\usepackage{amssymb}
\usepackage{stackrel}
\usepackage{graphicx}

\makeatletter
\numberwithin{equation}{section}
\numberwithin{figure}{section}
\theoremstyle{plain}
\newtheorem{thm}{\protect\theoremname}[section]
\theoremstyle{plain}
\newtheorem{cor}[thm]{\protect\corollaryname}
\theoremstyle{definition}
\newtheorem{defn}[thm]{\protect\definitionname}
\theoremstyle{remark}
\newtheorem{rem}[thm]{\protect\remarkname}
\theoremstyle{plain}
\newtheorem{lem}[thm]{\protect\lemmaname}
\theoremstyle{plain}
\newtheorem{prop}[thm]{\protect\propositionname}

\@ifundefined{date}{}{\date{}}
\usepackage{amsthm}
\usepackage{longtable}

\DeclareSymbolFont{extraup}{U}{zavm}{m}{n}
\DeclareMathSymbol{\varheart}{\mathalpha}{extraup}{86}
\DeclareMathSymbol{\vardiamond}{\mathalpha}{extraup}{87} 

\makeatother

\providecommand{\corollaryname}{Corollary}
\providecommand{\definitionname}{Definition}
\providecommand{\lemmaname}{Lemma}
\providecommand{\propositionname}{Proposition}
\providecommand{\remarkname}{Remark}
\providecommand{\theoremname}{Theorem}

\begin{document}
\global\long\def\C{\mathbb{C}}%
\global\long\def\Cd{\C^{\delta}}%
\global\long\def\Cprim{\C^{\delta,\circ}}%
\global\long\def\Cdual{\C^{\delta,\bullet}}%
\global\long\def\Od{\Omega^{\delta}}%
\global\long\def\Oprim{\Omega^{\delta,\circ}}%
\global\long\def\Odual{\Omega^{\delta,\bullet}}%

\global\long\def\P{\mathsf{\mathbb{P}}}%
 
\global\long\def\E{\mathsf{\mathbb{E}}}%
 
\global\long\def\sF{\mathcal{F}}%
 
\global\long\def\ind{\mathbb{I}}%

\global\long\def\R{\mathbb{R}}%
 
\global\long\def\Z{\mathbb{Z}}%
 
\global\long\def\N{\mathbb{N}}%
 
\global\long\def\Q{\mathbb{Q}}%

\global\long\def\C{\mathbb{C}}%
 
\global\long\def\Rsphere{\overline{\C}}%
 
\global\long\def\re{\Re\mathfrak{e}}%
 
\global\long\def\im{\Im\mathfrak{m}}%
 
\global\long\def\arg{\mathrm{arg}}%
 
\global\long\def\i{\mathfrak{i}}%
\global\long\def\eps{\varepsilon}%
\global\long\def\lamb{\lambda}%
\global\long\def\lambb{\bar{\lambda}}%

\global\long\def\D{\mathbb{D}}%
 
\global\long\def\H{\mathbb{H}}%
\global\long\def\F{\mathcal{F}}%
\global\long\def\outside{\mathcal{\text{out}}}%
\global\long\def\winding{\mathrm{wind}}%

\global\long\def\dist{\mathrm{dist}}%
 
\global\long\def\reg{\mathrm{reg}}%

\global\long\def\half{\frac{1}{2}}%
 
\global\long\def\sgn{\mathrm{sgn}}%
\global\long\def\Conf{\mathrm{Conf}}%
\global\long\def\ZLT{\mathrm{Z_{LT}}}%
\global\long\def\Wind{\mathrm{winding}}%
\global\long\def\Edges{\mathrm{\mathcal{E}}}%
\global\long\def\Dual#1{#1^{\bullet}}%

\global\long\def\bdry{\partial}%
\global\long\def\pa{\partial}%
 
\global\long\def\cl#1{\overline{#1}}%
\global\long\def\ZFK{Z_{\text{FK}}}%
\global\long\def\Clusters{\mathcal{C}}%
\global\long\def\Odp{\hat{\Omega}^{\delta}}%
\global\long\def\Edges{\mathcal{E}}%
\global\long\def\Vertices{\mathcal{V}}%
\global\long\def\bcond{\mathcal{\beta}}%
\global\long\def\free{\mathcal{\mathrm{free}}}%
\global\long\def\wired{\mathrm{wired}}%
\global\long\def\Medial{\mathcal{M}}%
\global\long\def\Edgesprim{\hat{\Edges}}%
\global\long\def\Conf{\mathrm{Conf}}%

\title{On multiple SLE for the FK--Ising model}
\author{Konstantin Izyurov}
\email{konstantin.izyurov@helsinki.fi}
\address{P.O. Box 68 (Pietari Kalmin katu 5), University of Helsinki, Finland}
\begin{abstract}
We prove convergence of multiple interfaces in the critical planar
$q=2$ random cluster model, and provide an explicit description of
the scaling limit. Remarkably, the expression for the partition function
of the resulting multiple SLE$_{16/3}$ coincides with the bulk spin
correlation in the critical Ising model in the half-plane, after formally
replacing a position of each spin and its complex conjugate with a
pair of points on the real line. As a corollary, we recover Belavin--Polyakov--Zamolodchikov
equations for the spin correlations. 
\end{abstract}

\maketitle

\section{Introduction}

Schramm--Loewner evolution \cite{Schramm_SLE} provides a geometric
description of scaling limits of critical planar models of statistical
mechanics. Its importance stems for the fact that SLE is characterized
by two simple properties, namely, the \emph{conformal invariance}
and the \emph{domain Markov property}. When the random curve in question
is described by Loewner evolution, these two properties imply that
the driving process has independent, identically distributed increments.
Since it is continuous, this identifies it as a Brownian motion with
a constant drift; mild additional symmetries such as scaling outrule
the latter.

This simple characterization requires the boundary conditions to be
sufficiently simple, so that any domain (in particular, the domain
slit by the initial segment of the curve) can be conformaly mapped
onto any other domain in such a way that the boundary conditions match.
This can be achieved when there are no more than three ``marked points''
on the boundary (i. e., points where boundary conditions change),
or one on the boundary and one in the bulk. Examples include a single
loop-erased random walk curve with Dirichlet boundary conditions \cite{LSW_LERW},
harmonic explorer \cite{SS_explorer} and the level lines of the Gaussian
Free Field \cite{SS_GFF} with jump boundary condtitions, Dobrushin
\cite{ChelkakSmirnov_et_al} and plus/minus/free \cite{Hongler_Kytola,IzyurovFree}
boundary conditions in the Ising model and wired/free boundary conditions
in the FK--Ising model \cite{ChelkakSmirnov_et_al}. These exmaples
lead to chordal, radial or dipolar SLE.

When the boundary conditions are more complicated, additional insight
is needed to characterize possible laws of the driving process. On
the physical level of rigor, it is clear that the law of the initial
segments of the curves should be absolutely continuous under the change
of boundary conditions far away. Hence, in general, the law of the
driving process should be given by a Brownian motion with a (time-dependent)
drift. Moreover, the Radon-Nikodym derivative with respect to an interface
with simpler (e. g., Dobrushin-type) boundary conditions can be written
as a ratio of partition functions, from which the drift term can be
derived by Girsanov transform. This led Bauer, Bernard, and Kytölä
\cite{bauer2005multiple} to a conjecture that to each boundary conditions
in a simply-connected domain $\Omega,$ one can associate an \emph{SLE
partition function} $Z(b^{(1)},\dots,b^{(n)})$, $b^{(1)},\dots,b^{(n)}\in\R,$
so that the driving process $b_{t}^{(1)}$ describing the curve $\gamma_{t}$
starting form $b^{(1)}$ satisfies the SDE 
\[
db_{t}^{(1)}=\sqrt{\kappa}dB_{t}+\kappa\partial_{b^{(1)}}\log Z(b_{t}^{(1)},\dots,b_{t}^{(n)})dt,
\]
 where $b_{t}^{(i)}=g_{t}(\varphi(b^{(i)}))$ for $i\geq2$, $g_{t}$
are the Loewner maps, $\varphi$ is a conformal map from $\Omega$
to the upper half-plane $\H$, and $b^{(i)}\in\partial\Omega$ are
the marked points for the boundary conditions in question. Moreover,
since $Z(b_{t}^{(1)},\dots,b_{t}^{(n)})$ can be identified with a
``boundary change operator'' correlation in the corresponding conformal
field theory, the function $Z$ was conjectured to satisfy a system
of second order partial differential equations known as Belavin--Polyakov--Zamolodchikov
(BPZ) equations in Conformal Field Theory \cite{BPZ}. Alternatively,
these equations can be derived from the fact that since $Z(b_{t}^{(1)},\dots,b_{t}^{(n)})$
is supposed to be a Radon--Nikodym derivative with respect to a chordal
SLE, it should be a (conformally covariant) chordal SLE martingale.

Making this reasoning rigorous is hard, since it requires controlling
the scaling limits of partition functions, in particular, in rough
domains. However, it was discovered by Dubédat \cite{dubedat_commut}
and independently by Zhan \cite{zhan2008duality} that if each of
the marked points $b^{(1)},\dots,b^{(n)}$ has its own SLE-like interface
growing from it, then natural consistency conditions, or ``commutation
relations'', actually imply the existence of an SLE partition function
with the above properties. 

Recently, a lot of progress has been made in finding relevant solutions
to the BPZ equations, or proving that the solutions with required
properties are unique. The upshot of these results is that for $2n$
marked points on the boundary, any multiple SLE is a mixture of one
of $\frac{(2n)!}{n!(n+1)!}$ \emph{pure geometry} multiple SLE, i.
e., the ones where marked points are connected to each other in a
prescribed planar pattern. The relevant description was proven by
Flores and Kleban in \cite{FlorKleb1,FlorKleb2,FlorKleb3,FlorKleb4}.
Independently, Kytölä and Peltola \cite{Kytola_peltola_quantum,Kytola_Peltola_pure}
have given explicit expressions for the partition function of pure
geometry multiple SLE in terms of Coulomb gas integrals for all $\kappa\notin\Q$.
The restriction to $\kappa\notin\Q$ is due to the fact that representation
theory of the quantum group $\mathcal{U}_{q}(su_{2})$, $q=e^{4\pi i/\kappa}$
is used in the construction, and this theory is much more intricate
for $q$ a root of unity. 

An independent line of study, started by Lawler and Kozdron \cite{Lawler_Kozdron},
bypasses completely the theory of BPZ equations. Instead, it purports
to construct pure geometry multiple SLE using Brownian loop measures,
and then to prove that there is at most one process satisfying a natural
set of axioms. This program has been recently completed by Beffara,
Peltola and Wu \cite{beffara2018uniqueness} in the case $\kappa\in(0;4]$.
For $\kappa\in(4;6]$, they got a corresponding result conditionally
on convergence of single interface in the corresponding random-cluster
model. In particular, since this convergence was established for $\kappa=\frac{16}{3}$,
their result implies convergence of multiple interfaces in the FK
model, conditioned on the connection geometry. Their result does not
yield explicit description of the law of the curves.

The above results mostly concern the pure geometry case. On the other
hand, in the underlying lattice model, there is usually a natural
``physical'' measure on the interfaces, without restrictions on
how they connect to each other. The corresponding multiple SLE partition
function (sometimes called the full partition function) can, in principle,
be found by analysis of the space of solution to BPZ equations; however,
deriving explicit expressions at rational $\kappa$ by this method
is difficult. When the convergence of interfaces is known, usually
integrability features of the model allow one to derive explicit solution
for the ``physical'' multiple SLE, as follows:
\begin{itemize}
\item In the critical \emph{Ising model} with alternating $+/-/\dots+/-$
boundary conditions, the interfaces converge \cite{IzyurovMconn,peltola2018crossing}
to multiple SLE$_{3}$ with partition function 
\[
Z(b^{(1)},\dots,b^{(2n)})=\mathrm{Pf}\left[(b^{(i)}-b^{(j)})^{-1}\right].
\]
The Pfaffian structure of the partition function is due to the fact
that boundary condition change can be achieved by placing a fermion
on the boundary. In \cite{IzyurovFree}, this has been extended to
allow free boundary arcs, in which case there is no simple Pfaffian
structure. For the the multiply-connected case, see \cite{IzyurovMconn}. 
\item In the \emph{Gaussian free field} with alternating $+\lambda/-\lambda/\dots/+\lambda/-\lambda$
boundary conditions, the level lines are multiple SLE$_{4}$ with
the partition function 
\[
Z(b^{(1)},\dots,b^{(2n)})=\prod_{i<j}(b^{(i)}-b^{(j)})^{\frac{1}{2}(-1)^{i-j}};
\]
which is in fact also the partition function of the underlying GFF
\cite{dubedat_GFF,Peltola_Wu_MSLE}. See also \cite{hagendorf2010gaussian,IzyurovKytola}
for the doubly-connected case.
\item The branches between $2n$ boundary points in \emph{Uniform spanning
tree} with wired boundary conditions converge to multiple SLE$_{2}$
with partition function 
\[
Z(b^{(1)},\dots,b^{(2n)})=\sum_{\omega}\sgn(\omega)\det\left[(b^{(i)}-b^{\omega(j)})^{-2}\right]_{i<\omega(i);j<\omega(j)}
\]
where the sum is over all involutions $\omega:\{1,\dots,2n\}\to\{1,\dots,2n\}$
without fixed points \cite{KarrilaKytolaPeltola,Karrila_Kytola_Peltola,karrila2020ust}.
The structure of the partition function is related to the determinantal
nature of the UST and the Fomin identity \cite{Fomin}.
\end{itemize}
The main contribution of this paper is the corresponding result for
the FK--Ising model.
\begin{thm}
\label{thm: main}The interfaces in the critical FK--Ising model
with free boundary conditions on $(b^{(1)},b^{(2)})$, $(b^{(3)},b^{(4)})$,
$\dots,$ $(b^{(2n-1)},b^{(2n)})$ and wired boundary conditions on
$(b^{(2)},b^{(3)}),$ $(b^{(4)},b^{(5)})$, $\dots$, $(b^{(2n)},b^{(1)})$
converge to multiple SLE$_{16/3}$ with partition function 
\begin{equation}
Z(b^{(1)},\dots,b^{(2n)})=\prod_{k=1}^{n}(b^{(2k)}-b^{(2k-1)})^{-\frac{1}{8}}\left(\sum_{\sigma\in\{\pm1\}^{n}}\prod_{i<j}\chi_{ij}^{\frac{\sigma_{i}\sigma_{j}}{4}}\right)^{\frac{1}{2}},\label{eq: partition_function}
\end{equation}
where $\chi_{ij}=\frac{(b_{2i}-b_{2j})(b_{2i-1}-b_{2j-1})}{(b_{2i}-b_{2j-1})(b_{2i-1}-b_{2j})}.$
\end{thm}

The mode of convergence is that the collections of full curves converge
to the corresponding global multiple SLE, see Definition \ref{def: global_msle}.
The technicalities are by now quite standard in the one-curve case,
where precompactness and similar issues have been resolved by Kemppainen
and Smirnov \cite{AnttiStas}. The multi-curve case they has been
recently systematically treated by Karrila \cite{karrila2019multiple},
who takes RSW-type bounds and the convergence in the mode of Lemma
\ref{lem: local_conv} below as inputs and explores conclusions. We
use some of his arguments, but our exposition is self-contained, only
relying on \cite{AnttiStas}.

The result of Theorem \ref{thm: main} was conjectured by Flores,
Simmons, Kleban and Ziff \cite{SimmonsFloresKlebanZiff}. Their conjecture
was based on the observation that the expression in (\ref{eq: partition_function})
formally coincides with the \emph{bulk spin correlation function}
in the Ising model on the upper half-plane when each pair of real
numbers $b^{(2i-1)},b^{(2i)}\in\R$ is replaced with a pair of complex
conjugates $a_{i},\bar{a}_{i}$. Since the spin correlations are believed
to satisfy the BPZ equations, so should (\ref{eq: partition_function}).
The spin correlations were rigorously computed in \cite{ChelkakHonglerIzyurov};
however, the author is not aware of a published proof that they do
indeed satisfy the BPZ equations (although the result was announced
in \cite{BurkhardtGuim}). We can actually derive this result from
Theorem \ref{thm: main} and Dubédat's commuting SLE theory:
\begin{cor}
The spin correlations in the scaling limit of the critical Ising model
in the half-plane, given by the formula 
\[
\prod_{k=1}^{n}(\im a_{i})^{-\frac{1}{8}}\left(\sum_{\sigma\in\{\pm1\}^{n}}\prod_{i<j}\left|\frac{a_{i}-a_{j}}{a_{i}-\bar{a}_{j}}\right|^{\frac{\sigma_{i}\sigma_{j}}{2}}\right)^{\frac{1}{2}}
\]
satisfy the BPZ equations; namely, for each $i=1,\dots,n$, they are
annihilated by 
\begin{equation}
\frac{8}{3}\frac{\partial^{2}}{(\partial a_{i})^{2}}+\sum_{j\neq i}\frac{2}{a_{j}-a_{i}}\frac{\partial}{\partial a_{j}}+\sum\frac{2}{\bar{a}_{j}-a_{i}}\frac{\partial}{\partial\bar{a}_{j}}+\sum_{j\neq i}\frac{1/8}{(a_{j}-a_{i})^{2}}+\sum\frac{1/8}{(\bar{a}_{j}-a_{i})^{2}}.\label{eq: BPZ}
\end{equation}
\end{cor}

\begin{proof}
The scaling limit of FK--Ising interfaces is, by construction, a
family of commuting SLE. By Dubédat's commutation relations \cite{dubedat_commut},
see \cite[Section 5, in particular (5.47)]{Graham} for an explicit
treatment, multiple SLE$_{16/3}$ partition function (\ref{eq: partition_function})
(which is determined uniquely up to a multiplicative constant by its
logarithmic derivatives, and thus by the law of the curves) satisfies
these equations with $(a_{i},\bar{a}_{i})$ replaced by $(b_{2i-1};b_{2i}).$
The result follows.
\end{proof}
The result of Theorem \ref{thm: main} for $n=1$ was established
in \cite{ChelkakSmirnov_et_al}, relying on the breakthrough proof
by Smirnov \cite{Smirnov_Ising,Smirnov_Towards,ChelkakSmirnov2} of
conformal invariance of fermionic observable, combined with the precompactness
results of \cite{AnttiStas} and Russo-Seymour-Welsh bounds \cite{duminil2011connection,chelkak2016crossing};
see also \cite{KemppainenTuisku} for the doubly connected case. For
$n=2$, the main ingredient was obtained in \cite{ChelkakSmirnov2},
where convergence of the martingale observable was proven, and the
details for the convergence of interfaces were given in \cite{AnttiStasHypergeo}.
These results were later used to describe full scaling limit of the
loop representation of the FK--Ising model \cite{KemppSmiBdryTouching,kemppainen2016conformal}.

In this paper, for simplicity, we work exclusively with the square
lattice. The results can be readily extended to isoradial setup by
the techniques of \cite{ChelkakSmirnov2}; recently, Chelkak has proven
the one-curve convergence result in a fully universal setup via s-embeddings
\cite{chelkak2017planar,chelkak2020dimer,chelkak2020universality}.
Eventually, it should be possible to derive the result of the present
paper in the same generality.

The ``wired'' boundary arcs in Theorem \ref{thm: main} are meant
to be wired altogether. Another natural setup, where they are wired,
but not wired to each other, is actually dual to this one, so we do
not need to consider it separately. It would be natural to consider
other external connections between the wired arcs. Given such a connection,
the Radon--Nikodym derivative of the corresponding multiple SLE with
respect to the one considered in Theorem \ref{thm: main} is simply
a function of \emph{connection pattern }of the multiple interfaces
inside the domain. Hence, calculating the law of the curves in this
situation is essentially equivalent to computing the probabilities
of all $\frac{(2n)!}{n!(n+1)!}$ connection patterns. While we are
at the moment unable to give explicit expressions for these probabilities,
the convergence of the interfaces implies conformal invariance.
\begin{cor}
\label{cor: connection_prob}Let $\Od$ be discrete domains with $2n$
marked boundary points $b^{(1,\delta)},\dots,b^{(2n,\delta)}$. Let
$\pi$ be a partition of the set $(b^{(2,\delta)},b^{(3,\delta)}),(b^{(4,\delta)},b^{(5,\delta)}),\dots$
of wired boundary arcs. Consider the critical FK--Ising model in
$\Od$ with boundary conditions as in Theorem \ref{thm: main}, and
let $\pi^{\delta}$ be the random partition of the set of wired boundary
arcs induced by the random clusters inside $\Od$. Then, for each
$\pi$, the quantity $\P(\pi^{\delta}=\pi)$ has a conformally invariant
scaling limit.
\end{cor}

In \cite{IzyurovFree}, it was noted that probabilities of a number
of connection events can be computed directly as special values discrete
holomorphic observables. This leads to a proof of their convergence
and conformal invariance in the scaling limit that completely bypasses
the SLE theory. In the half-plane, the expression for the scaling
limits are explicit quadratic irrational functions. The class of these
events was, however, described in \cite{IzyurovFree} incorrectly.
In fact, there are $2^{n-1}-1$ such non-trivial events, one for each
non-empty subset $S$ of the set of free boundary arcs with $|S|$
even. The event corresponding to $S$ can be described as ``no dual
cluster touches an odd number of arcs in $S$''. The explicit expression
are given by the ratios of the half-plane spin-disorder correlations
to spin correlations, with the familiar replacement $(a_{i},\bar{a}_{i})\to(b_{2i-1};b_{2i})$
and the disorders corresponding to arcs in $S$, see further details
in \cite{CHI_Mixed}.

In particular, when $|S|=2$, this is just the probability that two
wired arcs are connected, generalizing a result from \cite{ChelkakSmirnov2}.
We do not know whether all connection probabilities, or even any connection
probabilities other than just described, are given by quadratic irrational
functions.

The paper is organized as follows. In Section \ref{sec: The-FK-Ising-model},
we introduce the graph notation and define the model. In Section \ref{sec: The-martingale-observable},
we recall the definition of a discrete holomorphic observable and
the convergence result from \cite{IzyurovFree}, and show that this
observable also possesses a martingale property with respect to the
FK--Ising interface. In Section \ref{sec: Tightness}, we prove tightness
of the interfaces and show that the scaling limit of the martingale
observable is a martingale with respect to any sub-sequential scaling
limit of the interfaces. In Section \ref{sec: Proof}, we prove Theorem
\ref{thm: main}. The main idea of the derivation of the law of the
driving process (up to disconnection threshold) from the martingale
observable is as in \cite{zhanLERW,IzyurovKytola,IzyurovMconn}, and
is based on the expansion of the martingale observable at the tip
of the curve. The version of this argument presented here uses contour
integration and is significantly shorter as compared to \cite{IzyurovMconn};
of a separate interest is a short proof of Lemma \ref{lem: semimart}
showing that the driving process is a semi-martingale. After deriving
the convergence in a ``local'' mode, an easy application of the
RSW results of \cite{chelkak2016crossing} shows that if the discrete
interface almost disconnects the domain, then, with high probability,
it actually does disconnect it quickly and with ``nothing interesting''
happening in between. This allows to conclude the proof by induction.

The author is grateful to Alex Karrila, Eveliina Peltola and Hao Wu
for stimulating discussions.

\section{The FK--Ising model}

\label{sec: The-FK-Ising-model}We denote $\Cd:=\delta\Z^{2}$, and
$\Cdual:=\delta\Z^{2}+\frac{\delta}{2}+\frac{i\delta}{2}$, the square
lattice of mesh size $\delta$ and its dual, respectively. By a \emph{simply
connected discrete domain} $\Od$ we mean a domain whose boundary
$\partial\Od$ is a simple nearest-neighbor closed path in $\Cdual$.
We denote by $\Edges(\Od)$ and $\Vertices(\Od)$ the sets of edges
and vertices of $\Cd$ that lie in $\Od$, respectively.

The \emph{alternating} \emph{wired/free} \emph{boundary conditions}
$\bcond^{\delta}$ for a simply-connected domain $\Od$ are specified
by a partition of $\partial\Od$ into a ``wired'' part $\bcond_{\wired}^{\delta}:=\beta_{2}^{\delta}\cup\dots\cup\beta_{2n}^{\delta}$
and a free part $\bcond_{\free}^{\delta}=\beta_{1}^{\delta}\cup\dots\cup\beta_{2n-1}^{\delta}$,
where $\bcond_{i}^{\delta}=(b^{(i,\delta)};b^{(i+1,\delta)})$ are
boundary arcs, and $b^{(2n+1,\delta)}=b^{(1,\delta)},\dots,b^{(2n,\delta)}\in\partial\Od\cap\Vertices(\Cdual)$
are boundary points separating the arcs, in counterclockwise order.
Put 
\[
\Edgesprim(\Od)=\Edges(\Od)\cup\{e\in\Edges(\Cd):e\cap\bcond_{\wired}^{\delta}\neq\emptyset\},
\]
the set of edges of $\Cd$ that either belong to $\Od$, of cross
the ``wired'' part of the boundary. 

The main subject of this paper is the critical planar $q=2$ Fortuin--Kasteleyn
random cluster model, or FK--Ising model \cite{fortuin1972random,grimmett2004random}.
This is a random collection $E\subset\Edgesprim(\Od)$ of edges chosen
according to the probability measure 
\[
\P_{\Od,\bcond^{\delta}}(E)=\frac{1}{\ZFK}\left(\frac{p}{1-p}\right)^{|E|}2^{\Clusters(E)}.
\]
Here $\Clusters(E)$ is the number of clusters (connected components)
in $E$, where all vertices outside of $\Od$ are considered to belong
to the same cluster, and 
\[
\ZFK=\sum_{E}p^{|E|}(1-p)^{-|E|}2^{\Clusters(E)}
\]
 is the partition function. 

Given $e\in\Edges(\Cd)$, its dual edge $\Dual e\in\Edges(\Cdual)$
is the edge of the dual lattice that crosses $e.$ Given a configuration
$E\subset\Edges(\Od)$, we define its \emph{dual configuration} $\Dual E\subset\Edges(\Cdual)$
by 
\[
\Dual E:=\{\Dual e:\left(e\in\Edges(\Od)\text{ or }e\cap\partial\Od\neq\emptyset\right)\text{ and }\,e\notin E\}.
\]
Note that it particular, $\Dual E$ always contains all the edges
comprising $\bcond_{\free}^{\delta}$. It is not hard to see that
adding an edge to $E$ either disconnects two clusters of $\Dual E$,
or connects two clusters of $E$, but never both. Hence, $\Clusters(\Dual E)+|E^{\bullet}|-\Clusters(E)$
does not, in fact, depend on $E$, and the probability of the configuration
can be also written as 
\[
\P_{\Od,\bcond^{\delta}}(E)=\frac{1}{\Dual{\ZFK}}\left(\frac{2(1-p)}{p}\right)^{|\Dual E|}2^{\Clusters(\Dual E)}.
\]
The self-dual point of the model, given by the condition $p/(1-p)=2(1-p)/p$,
is also known to be its critical point; this result in fact holds
for any $q\geq1$ \cite{beffara2012self}. For the rest of the paper,
we set $p$ to its critical value $p_{c}=2-\sqrt{2}$.

Note that in $\Dual E$, the arcs $\beta_{1}^{\delta},\dots,\beta_{2n-1}^{\delta}$
are wired, but not wired together, which is the duality mentioned
in the introduction. 

The \emph{medial lattice} $\Medial(\Cd)$ is the square lattice whose
vertices are the midpoints of edges of $\Cd$. Given a configuration
$E$, an \emph{exploration interface }is a nearest-neighbor path on
$\Medial(\Cd)$ that turns by $\pm\frac{\pi}{2}$ at every step, and
never crosses edges of either $E$, or $\Dual E$, or $\text{\ensuremath{\Edges}(\ensuremath{\Cd}\ensuremath{\ensuremath{\setminus\Od}})}$,
transversally, see Figure \ref{fig: domain}. An exploration interface
$\gamma$ is completely determined by its starting (oriented) edge
and the configuration $E$; in its turn, its initial segment determines
the state of all edges whose midpoints it has visited, except possibly
for the last one. 

We will be interested in the statistics of the interface $\gamma^{\delta}=(\gamma_{0}^{\delta},\gamma_{1}^{\delta},\dots)$
starting in between $\beta_{\wired}^{\delta}$ and $\beta_{\free}^{\delta}$,
say at $b^{(1,\delta)}$. More precisely, we let $\gamma_{0}^{\delta}$
be an oriented edge of $\Medial(\Cd)$ that has $b^{(1,\delta)}$
on its right and a vertex of $\Cd\setminus\Od$ on its left. We say
the an edge of $\Edgesprim(\Od)$ is \emph{revealed} by $\gamma_{[0,t]}^{\delta}$
if its midpoint is an endpoint of one of the edges $\gamma_{0}^{\delta},\dots,\gamma_{t-1}^{\delta}.$
By induction, the medial edge $\gamma_{t}^{\delta}$ will always have
on its right a dual vertex connected to $\beta_{1}^{\delta}$ by a
sequence of edges of $\Dual E$ revealed by $\gamma_{[0,t]}^{\delta}$,
and on its left a primal vertex connected to $\Cd\setminus\Od$ by
a sequence of edges of $E$ revealed by $\gamma_{[0,t]}^{\delta}$.

By $\Od_{t},$ we denote the (random) domain obtained by removing
from $\Od$ all the vertices that are incident to the edges of $E$
revealed by $\gamma^{\delta}$ by time $t$. Although $\Od_{t}$ is
not necessarily connected, each of its connected components is simply
connected. The domain $\Od_{t}$ is naturally equipped with the boundary
conditions $\bcond_{t}^{\delta}$ that are free on $\bcond_{\free}^{\delta}$
and on all dual edges revealed to be in $\Dual E$ by time $t$, and
wired elsewhere. As long as $b^{(2,\delta)},\dots,b^{(2n,\delta)}$
are on the boundary of the same connected component of $\Od_{t}$,
on that component $\bcond_{t}^{\delta}$ are specified by $2n$ arcs,
with the additional marked points separating the free boundary arc
adjacent to $b^{(2,\delta)}$ and the wired boundary arc adjacent
to $b^{(2n,\delta)}.$ The conditional law of $E$ given $\gamma_{[0,t]}^{\delta}$
is the union of the edges of $E$ revealed by time $t$ and a critical
FK--Ising configuration in $\Omega_{t}^{\delta}$ with boundary conditions
$\bcond_{t}^{\delta}$. This property is clear from the definition
and is an instance of the \emph{domain Markov property}.
\begin{figure}

\includegraphics[width=0.6\paperwidth]{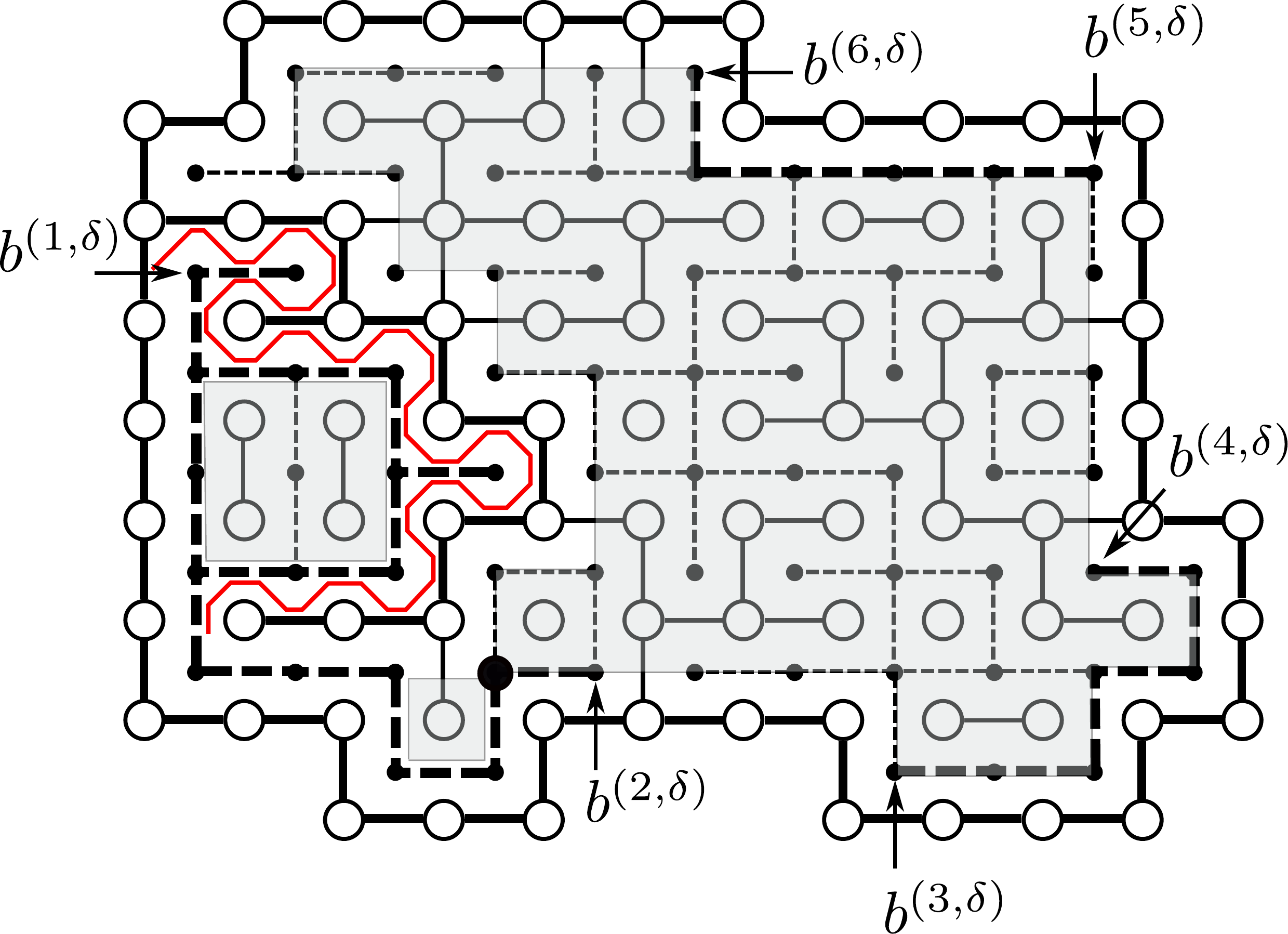}\caption{\label{fig: domain}An example of a discrete domain $\protect\Od$
with six marked boundary points, a configuration $E$ (solid edges
connecting the $\circ$ vertices of $\protect\Cd$) and $E^{\bullet}$
(dashed edges connecting the $\bullet$ vertices of $\protect\Cdual$).
The vertices of the exterior loop formed by solid edges belong to
$\protect\Cd\setminus\protect\Od$ and are thus wired altogether.
An initial segment $\gamma_{[0,t]}^{\delta}$ of the interface starting
from $b^{(1,\delta)},$ for $t=24,$ is drawn in red. Bold edges and
dual edges are discovered by the interface, or their state is prescribed
by boundary conditions. The state of thin edges given $\gamma_{[0,t]}^{\delta}$
is genuinely random. In gray is the domain $\protect\Od_{t}$, which
consists of three connected components. One of them has on its boundary
$b^{(2,\delta)},\dots,b^{(2n,\delta)}$ and exactly one additional
marked boundary point, the big black dot.}

\end{figure}

\section{The martingale observable}

\label{sec: The-martingale-observable}In this section, we recall
from \cite{IzyurovFree} the definition of a discrete holomorphic
observable that has been used to derive convergence of multiple interfaces
in the spin Ising model in the presence of free boundary arcs. It
turns out that \emph{the same observable} possesses a martingale property
with respect to the FK interface. The observable was given in \cite{IzyurovFree}
in terms of the low-temperature expansion. The proof of its martingale
property consists of first writing it in the order-disorder formalism
of Kadanoff--Ceva \cite{kadanoff1971determination}, see also \cite{ChelkakCimasoniKassel,CHI_Mixed},
and then invoking Edwards--Sokal coupling \cite{edwards1988generalization,grimmett2004random}
and the domain Markov property. 

For a discrete domain $\Od$, denote by $\Odual:=\Cdual\cap\left(\Od\cup\partial\Od\right)$
its dual domain. Let $a,z$ be two \emph{corners} in $\Od$ or adjacent
to its boundary, i. e., two midpoints of segments ($a^{\bullet}a^{\circ}$)
and $(z^{\bullet}z^{\circ})$ connecting vertices $a^{\circ},z^{\circ}\in\Vertices(\Cd)$
with adjacent dual vertices $a^{\bullet},z^{\bullet}\in\Odual$. We
denote by $\Conf_{\Od}$ the set of all subsets $S\subset\Edges(\Odual)$
such that all vertices of $\Odual$ have even degree in $S$, and
by $\Conf_{\Od}(a,z)$ the set of all subsets $S\subset\Edges(\Odual)$
such that all such that all vertices of $\Odual$ have even degree
in $S$, except for $a^{\bullet},z^{\bullet}$ that have odd degree
in $S$. The winding $\winding(S)$ is defined by the following procedure:
add to $S$ the segments $(a,a^{\bullet})$ and $(z^{\bullet},z)$,
and decompose the resulting graph into a collection of loops and a
path $S'$ connecting $a$ and $z$, in such a way that the loops
and the path do not intersect themselves or each other transversally.
(This amounts to resolving each vertex of degree four in one of the
two possible ways.) Then, the winding number of $S'$ (i. e., the
rotation number of its tangent vector) does not depend on the decomposition,
and that is defined to be $\winding(S).$ 
\begin{defn}
(\cite{IzyurovFree}, Section 2, case $m=0$, $s=2$) Let $(\Od,\bcond^{\delta})$
be a discrete simply connected domain, and let $a\notin\text{\ensuremath{\Od}}$
be a marked corner such that $a^{\bullet}\in\pa\Od$ separates $\beta_{\wired}^{\delta}$
from $\beta_{\free}^{\delta}$. The fermionic observable is defined
by 

\[
F_{\Od,\bcond^{\delta}}(z)=\i\eta_{a}\cdot\frac{\sum_{S\in\Conf_{\Od}(a,z)}e^{-\frac{\i}{2}\winding(S)}\alpha^{|S\setminus\beta_{\free}^{\delta}|}.}{\sum_{S\in\Conf_{\Od}}\alpha^{|S\setminus\beta_{\free}^{\delta}|}.},\quad\alpha:=\sqrt{2}-1,
\]
where $\eta_{a}:=\left(\frac{a-a^{\bullet}}{|a-a^{\bullet}|}\right)^{-\frac{1}{2}}.$ 
\end{defn}

\begin{rem}
The observable $F_{\Od,\bcond^{\delta}}(z)$ depends on the choice
of the square root in the definition of $\eta_{a}$. If a family of
simply-connected domains all have common part of the boundary, as
will be the case with domains $\Od_{t},$ then such a choice can be
made coherently by fixing the sign of the square root of the outer
normal at some point of the common boundary and then extending it
continuously, say, in counterclockwise direction. With this convention,
$F_{\Od,\bcond^{\delta}}$ depends only on $a^{\bullet}$ but not
on the choice of the corner $a$ adjacent to it. We incorporate the
choice of $a^{\bullet}$ into the boundary conditions $\bcond^{\delta}$
and do not stress it separately in the notation. 
\end{rem}

\begin{lem}
\label{lem: mart}For any corner $z\in\Od$, the sequence 
\[
F_{\Od_{t},\bcond_{t}^{\delta}}(z)
\]
 is a martingale with respect to the filtration $\F_{t}:=\sigma(\gamma_{[0,t]})$,
as long as $z$ is in the same connected component of $\Od_{t}$ as
$b^{(2,\delta)},\dots,b^{(2n,\delta)}$
\end{lem}

\begin{proof}
Recall the \emph{Edwards--Sokal coupling}: one samples $E\subset\Edgesprim(\Od)$
from the FK--Ising measure and then assigns a spin $\sigma=\pm1$
to each vertex uniformly at random subject to the conditions that
all vertices in each cluster receive the same spin. The resulting
spin configuration $\sigma:\Vertices(\Od)\to\left\{ \pm1\right\} $
has the distribution of the critical Ising model in $\Od$ with free
boundary conditions on $\beta_{\free}^{\delta}$ and fixed boundary
conditions on $\beta_{\wired}^{\delta}$ (i. e., the spins do interact
across $\beta_{\wired}^{\delta}$ and don't interact across $\bcond_{\free}^{\delta}$,
and all the spins outside of $\Od$ are required to be the same).
By domain Markov property, it is clear that conditionally on $\gamma_{[0,t]}$,
the spin configuration $\sigma$ has the distribution of the Ising
model in $\Od_{t}$ with the above boundary conditions. 

We now express $F_{\Od,\bcond^{\delta}}$ as an Ising model correlation.
Fix two lattice paths $\gamma^{\bullet},\gamma^{\circ}$ on $\Odual$
(respectively, $\Od$) connecting $z^{\bullet}$ to the free boundary
arc $\bcond_{1}^{\delta}$, and then, along $\bcond_{1}^{\delta}$,
to $a^{\bullet},$ (respectively, connecting $z^{\circ}$ to a point
outside $\Od$ and then, counterclockwise, to $a^{\circ}$). Then,
$S\mapsto S\triangle\gamma^{\bullet}$, where $\triangle$ stands
for symmetric difference, is a bijection between $\Conf_{\Od}(a,z)$
and $\Conf_{\Od}$. Consequently, we can write 
\[
F_{\Od,\bcond^{\delta}}(z)=\i\eta_{a}\cdot\frac{\sum_{S\in\Conf_{\Od}}e^{-\frac{\i}{2}\winding(S\triangle\gamma^{\bullet})}\alpha^{|(\gamma^{\bullet}\setminus\beta_{\free}^{\delta})\setminus S|-|(\gamma^{\bullet}\setminus\beta_{\free}^{\delta})\cap S|+|S\setminus\beta_{\free}^{\delta}|}}{\sum_{S\in\Conf_{\Od}}\alpha^{|S\setminus\beta_{\free}^{\delta}|}}.
\]
Using the low-temperature expansion, this can be written as 
\[
F_{\Od,\bcond^{\delta}}(z)=\i\eta_{a}\E_{\mathrm{Ising}}\left[e^{-\frac{\i}{2}\winding(S(\sigma)\triangle\gamma^{\bullet})}\prod_{(xy)\cap(\gamma^{\bullet}\setminus\bcond_{\free}^{\delta})\neq\emptyset}\alpha^{\sigma_{x}\sigma_{y}}\right].
\]
where $S(\sigma)\in\Conf_{\Od}$ is the set of dual edges separating
vertices with different spins. Now, we note that for a planar loop
$\gamma$, one has $e^{-\frac{\i}{2}\winding(\gamma)}=(-1)^{N(\gamma)+1}$,
where $N(\gamma)$ is the number of transversal self-intersections
of $\gamma$. Applying this to the (random) loop that is comprised
of $\hat{\gamma}^{\circ}:=(zz^{\circ})\cup\gamma^{\circ}\cup(a^{\circ}a)$
and the path from $a$ to $z$ in the decomposition of $S(\sigma)\triangle\gamma^{\bullet}$,
we infer that $e^{-\frac{\i}{2}\winding(S(\sigma)\triangle\gamma^{\bullet})}=e^{\frac{\i}{2}\winding(\hat{\gamma})}(-1)^{\gamma^{\bullet}\cap\gamma^{\circ}}(-1)^{|S(\sigma)\cap\gamma^{\circ}|+1}$,
where we take into account that any two planar loops have an even
number of transversal intersections, and the loops do not intersect
the random path. Finally, we note that $(-1)^{|S(\sigma)\cap\gamma^{\circ}|}=\sigma_{z^{\circ}}\sigma_{\mathrm{a^{\circ}}}$,
(recall that the spin is constant outside $\Od$), and $\i\eta_{a}e^{\frac{\i}{2}\winding(\hat{\gamma})}=\eta_{z}.$
We conclude that 
\[
F_{\Od,\bcond^{\delta}}(z)=\eta_{z}\E_{\mathrm{Ising},\Od}\left[\sigma_{z^{\circ}}\sigma_{a^{\circ}}\prod_{(xy)\cap(\gamma^{\bullet}\setminus\bcond_{\free}^{\delta})\neq\emptyset}\alpha^{\sigma_{x}\sigma_{y}}\right].
\]
Since the left-hand side does not depend on the choice of $\gamma^{\bullet}$,
neither does the right-hand side. Hence, as long as $z$ is in the
same connected components as $b^{(2,\delta)},\dots,b^{(2n,\delta)}$,
we may assume that $\gamma^{\bullet}\setminus\bcond_{\free}^{\delta}$
lies in $\Od_{t}$. By the above remark on the domain Markov property,
$F_{\Od_{t},\bcond_{t}^{\delta}}(z)=\E\left[F_{\Od,\bcond^{\delta}}(z)|\gamma_{[0,t]}\right]$,
and thus it is trivially a martingale.
\end{proof}
It was proven in \cite[Theorem 2.6]{IzyurovFree} (see also \cite{CHI_Mixed}
for a more general setup) that if $(\Od,\bcond^{\delta})\stackrel{\mathrm{Cara}}{\longrightarrow}(\Omega,\bcond)$,
then the observable $2^{-\frac{1}{4}}\pi^{\frac{1}{2}}\delta^{-1/2}F_{\Od,\bcond^{\delta}}(z)$
(more precisely, its sum over two corners adjacent to the same edge,
but to different vertices and dual vertices) converges in the scaling
limit to a holomorphic function $f_{\Omega,\bcond}(z)$ that satisfies,
under conformal maps, the covariance rule 
\begin{equation}
f_{\Omega,\bcond}(z)=\varphi'(z)^{\frac{1}{2}}f_{\varphi(\Omega),\varphi(\bcond)}(\varphi(z)).\label{eq: ccov}
\end{equation}
Moreover, if $\Omega=\H$ and $b^{(1)}<\dots<b^{(2n)}$, then the
observable is given by 
\[
f_{\H,\bcond}(z)=\frac{P_{\bcond}(z)}{\prod_{i=1}^{n}\sqrt{(z-b^{(2i-1)})(z-b^{(2i)})}},
\]
 where $P_{\bcond}$ is a polynomial of degree $\leq n-1$ whose coefficients
are uniquely determined by the following system of linear equations:
for $i=2,\dots,n$, one has 
\begin{align}
\lim_{z\to b^{(2i-1)}}f_{\H,\bcond}(z)\sqrt{(z-b^{(2i-1)})(z-b^{(2i)})}= & -\lim_{z\to b^{(2i)}}f_{\H,\bcond}(z)\sqrt{(z-b^{(2i-1)})(z-b^{(2i)})}\label{eq: f_linear_1}
\end{align}
for $i=2,\dots,n$, and 
\begin{equation}
\lim_{z\to b^{(1)}}\sqrt{z-b^{(1)}}f_{\H,\bcond}(z)=\i.\label{eq: f_linear_2}
\end{equation}

\begin{rem}
One has to make several minor adjustments to bring the results of
\cite{IzyurovFree} into the above form. First, we re-number the boundary
points so that the arc $(b_{2k-1},b_{2k})$ of \cite{IzyurovFree}
becomes $(b^{(1)},b^{(2)});$ we then have $a_{1}=b^{(1)}$ and $a_{2}=b^{(2)}$.
Second, the normalization factor in \cite[Theorem 2.6]{IzyurovFree}
is actually $\sum_{S\in\Conf_{\Od}(b^{(1)},b^{(2)})}\alpha^{|S\setminus\bcond_{\free}^{\delta}|}$
rather than $\sum_{S\in\Conf_{\Od}}\alpha^{|S\setminus\bcond_{\free}^{\delta}|}.$
These two are equal because $S\mapsto S\triangle(b^{(1)},b^{(2)})$
is a bijection between $\Conf_{\Od}$ and $\Conf_{\Od}(b^{(1)},b^{(2)}),$
which is weight preserving since $(b^{(1)},b^{(2)})\subset\bcond_{\free}^{\delta}.$
Finally, the result in \cite{IzyurovFree} gives normalizing condition
(\ref{eq: f_linear_2}) at $b^{(2)}$ rather than $b^{(1)}$; this
is equivalent to (\ref{eq: f_linear_2}) since in the course of the
proof of Theorem 2.6, the relation (\ref{eq: f_linear_1}) was proven,
without the $-$ sign, also for $i=1$.

In principle, one could write downs a (rather involved) explicit expression
of $f_{\H,\bcond},$ see \cite{CHI_Mixed}. However, all we shall
need is (\ref{eq: ccov}) and the following expansion:
\end{rem}

\begin{prop}
\label{prop: f_expansion}We have
\begin{equation}
f_{\H,\bcond}(b^{(1)},z)=\i\cdot(z-b^{(1)})^{-\frac{1}{2}}\left(1+2\mathcal{A}_{\bcond}\cdot(z-b^{(1)})+o(z-b^{(1)})\right)\quad z\to b^{(1)},\label{eq: f_expansion}
\end{equation}
where 
\[
\mathcal{A}_{\bcond}=2\partial_{b^{(1)}}\log Z\left(b^{(1)},\dots,b^{(2n)}\right)
\]
 and $Z\left(b^{(1)},\dots,b^{(2n)}\right)$ is given by (\ref{eq: partition_function}).
\end{prop}

\begin{proof}
The linear system (\ref{eq: f_linear_1}) -- (\ref{eq: f_linear_2})
is formally the same as the system (A2) in \cite[Appendix 2]{ChelkakHonglerIzyurov};
with $b^{(2i-1)};b^{(2i)}$ replacing $a_{i};\bar{a}_{i}$. Hence,
its solution is the same, and in fact $\mathcal{A}_{\bcond}$ is given
by the expression for $\mathcal{A}_{[\H,a_{1},\dots,a_{n}]}$ with
the above substitution.
\end{proof}

\section{Tightness and the martingale property in the scaling limit}

\label{sec: Tightness}We start by clarifying the conditions of Theorem
\ref{thm: main}. We assume that the discrete domains $\Od$ converge
to a simply-connected domain $\Omega$ in the sense of Carathéodory,
that $b^{(1)},\dots,b^{(2n)}\in\partial\Omega$ are boundary points
(degenerate prime ends), and that $b^{(1,\delta)},\dots,b^{(2n,\delta)}\in\partial\Od$,
as above, converge to $b^{(1)},\dots,b^{(2n)}$ respectively. In order
to avoid the situation of $b^{(i,\delta)}$ being inside a deep fjord
of $\Od$ that disappears in $\Omega$, forcing the initial segment
of the corresponding interface to wiggle outside $\Omega$, we need
to impose a regularity assumption on the approximations near $b^{(i)}.$
It is actually convenient to state this assumption in terms of the
behavior of the interface, namely, we require the for any $\eps>0,$
there is a neighborhood $U^{\eps}$ of $b^{(i)}$ in $\Omega$ and
a sequence of neighborhoods $U^{\eps,\delta}\stackrel{\text{Cara}}{\longrightarrow}U^{\eps}$
such that 
\begin{equation}
\P\left(\mathrm{diam}\left(\gamma_{[0,T^{\eps,\delta}]}^{(i,\delta)}\right)>\eps\right)<\eps\label{eq: bdry_reg}
\end{equation}
for all $\delta$ small enough. Here $\gamma_{[0,t]}^{(i,\delta)}$
is the initial segment of the interface starting at $b^{(i,\delta)}$,
and $T^{\eps,\delta}$ is the exit time from $U^{\eps,\delta}$. It
is clear that one can enforce this property by a suitable geometric
condition. 

Let $N$ be the (random) number of steps after which the interface
$\gamma_{t}^{\delta}$ starting at $b^{(1,\delta)}$ first exits the
domain, by which we mean that $\gamma_{N}^{\delta}\notin\Od$ and
the medial edges $\gamma_{N}^{\delta}$ has one of $b^{(i,\delta)}$
on its right; for topological reasons, this (random) index $i$ is
even, and we denote it by $I$. We have the following lemma: 
\begin{lem}
\label{lem: tightness}The family of random curves $\gamma_{[0,N]}^{\delta}$
is tight in the space of continuous planar curves taken up to re-parametrization.
Moreover, conditionally on $I$, any of its sub-sequential limits,
mapped to the upper half-plane $\H$ by a conformal map that sends
$b^{(I)}$ to infinity, is almost surely fully described by its chordal
Loewner chain, which has a continuous driving process. 
\end{lem}

\begin{proof}
Similarly to the argument for one curve given in \cite{ChelkakSmirnov_et_al},
we rely on \cite{AnttiStas}. The only difference is that in \cite{AnttiStas},
the target point is prescribed. Clearly, it suffices to prove the
tightness of the conditional laws of $\gamma_{[0,N]}^{\delta}$ given
$I$. Karrila \cite[Lemma 4.5]{karrila2019multiple} has shown that
in general, the conditions in \cite{AnttiStas} are not affected by
conditioning on the target; below we more or less follow his proof. 

By \cite[Theorem 1.5 and Corollary 1.8]{AnttiStas}, see also \cite{karrila2018limits},
it suffices to prove a uniform upper bound on the probability (given
$I$) of a crossing of a topological rectangle of modulus $\geq M$
with two opposite sides on the boundary that does not disconnect $b^{(1,\delta)}$
from $b^{(I,\delta)},$ with $M$ an absolute constant. 

Let $Q^{\delta}$ be such a rectangle, for definiteness, let $Q^{\delta}\cap\pa\Od\subset(b^{(1,\delta)},b^{(I,\delta)})$,
and split it into two rectangles $Q_{L}^{\delta},Q_{R}^{\delta}$
of moduli $M/2$, such that $\gamma_{[0,N]}^{\delta}$ crosses $Q_{R}^{\delta}$
only if it crosses $Q_{L}^{\delta}$ and it crosses $Q^{\delta}$
iff it crosses both. Let $\gamma_{t}^{i,\delta},$ $i=3,\dots,2n-1$
be the interface started at $b^{(i,\delta)}$ and stopped at its corresponding
$N_{i}$. Let $\Omega_{1}^{\delta}$ be the connected component of
$\Od\setminus\left(\gamma_{[0,N_{3}]}^{3,\delta}\cup\dots\cup\gamma_{[0,N_{2n-1}]}^{2n-1,\delta}\right)$
that has $b^{(1,\delta)}$ on its boundary. If $\gamma_{[0,N]}^{\delta}$
crosses $Q_{R}^{\delta}$, then there is an open FK percolation path
crossing of $Q_{R}^{\delta}$ inside $\Omega_{1}^{\delta}$. 

Let $A$ denote the event that none of $\gamma_{[0,N_{i}]}^{i,\delta}$
with $i>I$ crosses $Q_{L}^{\delta}$. Conditionally on $\Omega_{1}^{\delta}$,
the configuration in $\Omega_{1}^{\delta}$ is that of FK--Ising
model with free boundary conditions on the sub-arc $(b^{(1,\delta)},b^{(I,\delta)})$
of $\partial\Omega_{1}^{\delta}$ and wired boundary conditions on
$(b^{(I,\delta)},b^{(1,\delta)}).$ On $A$, any part of $\pa\Od_{1}$
that intersects $Q_{R}^{\delta}$ has free boundary conditions. Therefore,
by monotonicity in the boundary conditions, $\P(A\text{ and }\gamma_{[0,N]}^{\delta}\;\mathrm{crosses}\;Q_{R}^{\delta}|\Od_{1})$
is smaller than the probability, for the FK model in $Q_{R}^{\delta}$
with plus boundary conditions, to have an open path crossing $Q_{R}^{\delta}$.
By RSW bound \cite{duminil2011connection}, this probability is smaller
than $\frac{1}{4}$ if $M$ is large enough. 

Let $\Omega_{2}^{\delta}$ be the union of connected components of
$\Od\setminus\gamma_{[0,N]}^{\delta}$ that have $b^{(I+1,\delta)},\dots,b^{(2n,\delta)}$
on their boundaries. The event $A^{c}$ implies that there is a crossing
of $Q_{L}^{\delta}$ by dual-open edges in $\Omega_{2}^{\delta}$.
Given $\gamma_{[0,N]}^{\delta}$, the law of the model in $\Omega_{2}^{\delta}$
is that of the FK--Ising model with wired boundary conditions on
$\left(\partial\Od_{2}\setminus\partial\Od\right)\cup\bcond_{\wired}^{\delta}$
and free on $\bcond_{\free}^{\delta}$. In particular, any part of
$\pa\Od_{2}$ that intersects $Q_{L}^{\delta}$ carries wired boundary
conditions, and, by monotonicity and RSW again, we conclude that $\P(A^{c}|\gamma_{[0,N]}^{\delta})<\frac{1}{4}$
if $M$ is large enough. Since $I$ is measurable both with respect
to $\Od_{1}$ and $\gamma{}_{[0,N]}^{\delta},$ we conclude 
\[
\P\left[\gamma_{[0,N]}^{\delta}\;\mathrm{crosses}\;Q_{R}^{\delta}|I\right]\leq\P\left[A\text{ and }\gamma_{[0,N]}^{\delta}\;\mathrm{crosses}\;Q_{R}^{\delta}|I\right]+\P\left[A^{c}|I\right]\leq\frac{1}{2}.
\]
\end{proof}
\begin{rem}
The reader may notice a little subtlety in that in order to apply
the results of \cite{AnttiStas}, $M$ must be independent of the
domain, and in particular of the probabilities $\P(I=j).$

We now turn to the identification of the scaling limit. To this end,
we fix a conformal map $\varphi:\Omega\to\H$. We assume that $\varphi$
maps \textbf{$b^{(1)},\dots,b^{(2n)}$} to $b_{0}^{(1)}<\dots<b_{0}^{(2n)}\in\R$.
Fix any cross-cut $\omega$ in $\H$ that separates $b_{0}^{(1)}$
from $b_{0}^{(2)},\dots,b_{0}^{(2n)},\infty$. We let $\gamma$ be
any sub-sequential limit of the law of $\gamma_{[0,N]}^{\delta}$.
Parametrize $\gamma$ by half-plane capacity of $\gamma_{[0,t]}^{\H}:=\varphi(\gamma_{[0,t]})$,
which is possible at least until $t=T_{\omega}:=\min\{t:\gamma_{[0,t]}^{\H}\cap\omega\neq\emptyset\}.$
Let $\H_{t}$ be the unbounded connected component of $\H\setminus\gamma_{[0,t]}^{\H}$
and let $g_{t}:\H_{t}\to\H$ be the Loewner maps, which satisfy the
Loewner's equation 
\[
\partial_{t}g_{t}(z)=\frac{2}{g_{t}(z)-b_{t}^{(1)}},
\]
where $b_{t}^{(1)}$ is the random driving function. We denote $b_{t}^{(i)}:=g_{t}(b_{0}^{(i)}).$
The domain $\H_{t}$ comes with natural boundary conditions, changing
from wired to free at $\gamma_{t}^{\H}$, $b_{0}^{(3)},\dots,b_{0}^{(2n-1)}$
and back at $b_{0}^{(2)},\dots,b_{0}^{(2n)}$, and we denote by $\bcond_{t}$
the push-forward of these boundary conditions to $\H$ by $g_{t}$.
Thus, we have 
\begin{equation}
f_{\H_{t}}(z)=f_{\H,\bcond_{t}}(g_{t}(z))g_{t}'(z)^{\frac{1}{2}}.\label{eq: coov_UHP}
\end{equation}
for the scaling limit $f$ of the martingale observable as defined
in Section \ref{sec: The-martingale-observable}. A crucial consequence
of the discrete martingale property (Lemma \ref{lem: mart}) and the
convergence result is the following lemma:
\end{rem}

\begin{lem}
For each $z\in\H$ separated from $b_{0}^{(1)}$ by $\omega$, the
process $f_{\H_{t\wedge T_{\omega}}}(z)$ is a martingale. 
\end{lem}

\begin{proof}
This is a standard argument featuring e. g. in \cite{werner2007lectures}.
We may assume, by Skorokhod representation theorem, that the interfaces
$\gamma_{[0,N]}^{\delta}$ are all defined on the same probability
space and converge almost surely to a random curve $\gamma_{t}\subset\Omega$.
Define $\hat{\omega}:=\varphi^{-1}(\omega)$ so that $w:=\varphi^{-1}(z)$
is separated from $b^{(1)}$ by $\hat{\omega}$. For convenience,
we re-parametrize $\gamma_{t}$ and the discrete interfaces $\gamma_{t}^{\delta}$
by the conformal radius of the connected component of their complement
containing $w$, in $\Omega$ and $\Od$ respectively. We define $\tau_{\omega}\geq T_{\omega}$
to be a continuous modification of the hitting time of $\hat{\omega}$,
as in \cite[Appendix B]{karrila2019multiple}. We moreover define
$\tau_{\omega}^{\delta}$ to be a stopping time with respect to the
filtration $\F(\gamma_{[0,t]}^{\delta})$ on discrete curves that
converges to $\tau_{\omega}$ almost surely, which can be achieved
by a similar construction. 

We claim that, for any $t$, we have $F_{\Omega_{t\wedge\tau_{\omega}^{\delta}}^{\delta}}(w)\to f_{\Omega_{t\wedge\tau_{\omega}}}(w)$
almost surely. Indeed, on the complementary event, we can extract
a subsequence of $\delta$ along which 
\begin{equation}
\left|F_{\Omega_{t\wedge\tau_{\omega}^{\delta}}^{\delta}}(w)-f_{\Omega_{t\wedge\tau_{\omega}}}(w)\right|\geq c>0.\label{eq: contrad}
\end{equation}
 From that subsequence, by compactness, we can extract further subsequence
such that $\Omega_{t\wedge\tau_{\omega}^{\delta}}^{\delta}$ converges
in the sense of Carathéodory, and, moreover, the boundary conditions
$\bcond_{t}^{\delta}$ converge (i. e., the points $b_{t}^{(1,\delta)}$
on $\pa\Omega_{t\wedge\tau_{\omega}^{\delta}}^{\delta}\setminus\pa\Omega^{\delta}$
separating the wired arc from the free arc converge). It is then clear
that almost surely, the limit $\Omega_{t\wedge\tau_{\omega}^{\delta}}^{\delta}$
is given by $\Omega_{t\wedge\tau_{\omega}}.$ Also, almost surely,
$\lim b_{t}^{(1,\delta)}=\gamma_{t}$; on the complementary event,
$\gamma_{t}$ would have traversed part of the boundary in zero time
which we know has probability zero. But now the convergence result
of \cite[Theorem 2.6]{IzyurovFree} yields that (\ref{eq: contrad})
is impossible.

Let $H$ be any bounded, continuous function of the following data:
a simply-connected domain $\Omega$ equipped with a point $w\in\Omega$
and boundary conditions $\bcond$ defined by $2n$ marked points $b^{(1)},\dots,b^{(2n)}$
as above. Using that by compactness, all functions under the expectation
are bounded, we have 
\begin{multline*}
\E\left[f_{\Omega_{t\wedge\tau_{\omega}}}(w)H(\Omega_{s\wedge T_{\omega}})\right]=\lim_{\delta\to0}\E\left[F_{\Omega_{t\wedge\tau_{\omega}^{\delta}}^{\delta}}(w)H(\Omega_{s\wedge\tau_{\omega}^{\delta}}^{\delta})\right]\\
=\lim_{\delta\to0}\E\left[\E\left[F_{\Omega_{t\wedge\tau_{\omega}^{\delta}}^{\delta}}(w)|\Omega_{s\wedge\tau_{\omega}^{\delta}}^{\delta}\right]H(\Omega_{s\wedge\tau_{\omega}^{\delta}}^{\delta})\right]\\
=\lim_{\delta\to0}\E\left[F_{\Omega_{s\wedge\tau_{\omega}^{\delta}}^{\delta}}(w)H(\Omega_{s\wedge\tau_{\omega}^{\delta}}^{\delta})\right]=\E\left[f_{\Omega_{s\wedge\tau_{\omega}}}(w)H(\Omega_{s\wedge\tau_{\omega}})\right],
\end{multline*}
proving that $f_{\Omega_{t\wedge\tau_{\omega}}}(w)$ is a martingale
with respect to $\F(\Omega_{t\wedge\tau_{\omega}})$ and therefore
$f_{\H_{t\wedge\tau_{\omega}}}(z)=\varphi'(w)^{-\frac{1}{2}}f_{\Omega_{t\wedge\tau_{\omega}}}(w)$
is a martingale.
\end{proof}
It is clear from the explicit description of $f_{\H,\bcond}(z)$ that
it is continuous in $z$ and $\bcond$; hence the martingale from
the last Lemma is jointly continuous in $t$ and $z$.

\section{Proof of the main theorem.\label{sec: Proof}}
\begin{lem}
\label{lem: semimart}For any sub-sequential scaling limit $\gamma_{t}$
of the interface $\gamma_{t}^{\delta}$ , the driving process $b_{t\wedge T_{\omega}}^{(1)}$
is a (continuous) semi-martingale.
\end{lem}

\begin{proof}
Let $\omega_{1}$ be a cross-cut in the upper half-plane such that
$\omega_{1}\cup\overline{\omega_{1}}$ is a loop encircling $b_{0}^{(1)}$
and $\omega$, but no other marked points. Using (\ref{eq: f_expansion})
and Schwartz reflection, we can write, for $t\leq T_{\omega},$ 
\begin{multline*}
b_{t}^{(1)}=-\frac{1}{\pi i}\int_{g_{t}(\omega_{1})}wf_{\H,\bcond_{t}}^{2}(w)dw=-\frac{1}{\pi i}\int_{\omega_{1}}g_{t}(z)f_{\H,g_{t}(\bcond_{t})}^{2}(g_{t}(z))g_{t}'(z)dz\\
=-\frac{1}{\pi i}\int_{\omega_{1}}g_{t}(z)f_{\H_{t}}^{2}(z)dz.
\end{multline*}
Clearly, since $f_{\H_{t}}(z)$ is a continuous martingale, $g_{t}(z)f_{\H_{t}}^{2}(z)$
is a semi-martingale for each $z$; indeed, by Itô's formula, it is
a sum of the martingale 
\begin{equation}
\int_{0}^{t\wedge T_{\omega}}g_{s}(z)2f_{\H_{s}}(z)df_{\H_{s}}(z)\label{eq: martingale_part}
\end{equation}
 and the adapted bounded variation process 
\begin{equation}
\int_{0}^{t\wedge T_{\omega}}\partial_{s}g_{s}(z)f_{\H_{s}}^{2}(z)ds+\int_{0}^{t\wedge T_{\omega}}g_{s}(z)[df_{\H_{s}}(z)]^{2}.\label{eq: adapted_bvp}
\end{equation}
Integrating (\ref{eq: martingale_part}) and (\ref{eq: adapted_bvp})
in $z$ yields a martingale and an adapted bounded variation process
respectively, hence $b_{t}^{(1)}$ is a semi-martingale.
\end{proof}
\begin{lem}
\label{lem: local_conv}For any sub-sequential scaling limit $\gamma_{t}$
of the interface $\gamma_{t}^{\delta}$ , there is a Brownian motion
$B_{t}$ such that, for any cross-cut $\omega$ separating $b_{0}^{(1)}$
from $b_{0}^{(2)},\dots,b_{0}^{(2n)},\infty$, one has
\begin{equation}
b_{t\wedge T_{\omega}}^{(1)}=\sqrt{\frac{16}{3}}B_{t\wedge T_{\omega}}+\frac{16}{3}\int_{0}^{t\wedge T_{\omega}}\partial_{b^{(1)}}\log Z(b_{t}^{(1)},\dots,b_{t}^{(2n)})\,dt,\label{eq: driving_process}
\end{equation}
where the partition function $Z$ is given by (\ref{eq: partition_function}).
\end{lem}

\begin{proof}
Let $H_{k}(t,z):=g_{t}'(z)^{\frac{1}{2}}(g_{t}(z)-b_{t}^{(1)})^{k+\frac{1}{2}}.$
A straightforward application of Itô's formula, which is valid by
Lemma \ref{lem: semimart}, yields 
\begin{multline}
dH_{k}(t,z)=\left(2kdt+\left(\frac{k^{2}}{2}-\frac{1}{8}\right)\left[db_{t}^{(1)}\right]^{2}\right)H_{k-2}(t,z)-\left(k+\frac{1}{2}\right)H_{k-1}(t,z)db_{t}^{(1)}.\label{eq: H_k_exp}
\end{multline}
For a cross-cut $\omega_{1}$ as in the proof of Lemma \ref{lem: semimart},
using (\ref{eq: ccov}) and (\ref{eq: f_expansion}), we get 
\begin{multline*}
\int_{\omega_{1}}f_{\H_{t}}(z)H_{k}(t,z)dz=\int_{\omega_{1}}f_{\H,\bcond_{t}}(g_{t}(z))(g_{t}(z)-b_{t}^{(1)})^{k+\frac{1}{2}}g_{t}'(z)dz\\
=\int_{g_{t}(\omega_{1})}f_{\H,\bcond_{t}}(w)(w-b_{t}^{(1)})^{k+\frac{1}{2}}dw=\begin{cases}
0, & k\geq0;\\
-\pi, & k=-1;\\
-2\pi\mathcal{A}_{\bcond_{t}}, & k=-2,
\end{cases}
\end{multline*}
Applying the Itô formula to this identity, and using (\ref{eq: H_k_exp})
yields, for $k\geq-1,$ 
\begin{multline}
0=\int_{\omega_{1}}df_{\H_{t}}(z)H_{k}(t,z)dz-\left(k+\frac{1}{2}\right)db_{t}^{(1)}\int_{\omega_{1}}f_{\H_{t}}(z)H_{k-1}(t,z)\\
+\left(2kdt+\left(\frac{k^{2}}{2}-\frac{1}{8}\right)\left[db_{t}^{(1)}\right]^{2}\right)\int_{\omega_{1}}f_{\H_{t}}(z)H_{k-2}(t,z)dz\\
-\left(k+\frac{1}{2}\right)\int_{\omega_{1}}[df_{\H_{t}}(z);db_{t}^{(1)}]\cdot H_{k-1}(t,z)dz.\label{eq: coeff_expansion}
\end{multline}
Take the cross-variation with $b_{t}^{(1)}$ and note that the last
two terms do not contain the Brownian part. We obtain
\[
\int_{\omega_{1}}[df_{\H_{t}}(z);db_{t}^{(1)}]\cdot H_{k}(t,z)dz=\begin{cases}
0, & k\geq1;\\
-\frac{\pi}{2}\left[db_{t}^{(1)}\right]^{2}, & k=0;\\
\pi\mathcal{A}_{\bcond_{t}}\left[db_{t}^{(1)}\right]^{2}, & k=-1.
\end{cases}
\]
 Specializing (\ref{eq: coeff_expansion}) to $k=1$ now yields 
\[
\int_{\omega_{1}}df_{\H_{t}}(z)H_{1}(t,z)dz=\pi\left(2dt-\frac{3}{8}\left[db_{t}^{(1)}\right]^{2}\right),
\]
 and since $f_{\H_{t}}(z)$ is a martingale, we get $\left[db_{t}^{(1)}\right]^{2}=\frac{16}{3}dt.$
Plugging $k=0$ into (\ref{eq: coeff_expansion}) gives
\[
\int_{\omega_{1}}df_{\H_{t}}(z)H_{0}(t,z)dz=-\frac{\pi}{2}db_{t}^{(1)}+\frac{\pi}{4}\cdot\mathcal{A}_{\bcond_{t}}\left[db_{t}^{(1)}\right]^{2}.
\]
Since since $f_{\H_{t}}(z)$ is a martingale, we get that $b_{t\wedge T_{\omega}}^{(1)}-\frac{8}{3}\int_{0}^{t\wedge T_{\omega}}\mathcal{A}_{\bcond_{t}}dt$
is a martingale. Being a martingale with quadratic variation equal
to $\frac{16}{3}t\wedge T_{\omega}$ identifies it uniquely as the
Brownian motion $\sqrt{\frac{16}{3}}B_{t\wedge T_{\omega}}.$ Taking
into account Proposition \ref{prop: f_expansion} concludes the proof.
\end{proof}
\begin{rem}
The reader who finds the above proof cryptic may think about it as
taking the Itô derivative of $f_{\H_{t}}=f_{\H,\bcond_{t}}(g_{t}(z))g_{t}'(z)^{\frac{1}{2}}$
by differentiating the expansion (\ref{eq: f_expansion}), which becomes
an expansion in $H_{k}(t,z)$, term by term, and concluding that since
the resulting expression is drift-less, so must be the coefficients
of the expansion. 

We now explaining what happens after the time the interface crosses
all possible cross-cuts $\omega$. Let 
\[
\tau:=\sup\{t:\exists\omega\text{ separating }\text{\ensuremath{\gamma_{[0,t]}^{\H}}}\text{ from }b_{0}^{(2)},\dots,b_{0}^{(2n)},\infty\}.
\]
\end{rem}

\begin{lem}
\label{lem: hitting point}The limit $\gamma_{\tau}:=\lim_{t\to\tau}\gamma_{t}\in\partial\Omega$
exists and we have $\gamma_{\tau}\in(b^{(2)},b^{(2n)})\setminus\{b^{(3)},\dots,b^{(2n-1)}\}$
almost surely. 
\end{lem}

\begin{proof}
The existence of the limit is clear from the fact that $\gamma_{t}$
is a continuous curve. Clearly, we cannot have $\gamma_{\tau}\in\Omega$
or $\gamma_{\tau}\in(b^{(2n)},b^{(2)}),$ since in that case, $\tau$
would not be the supremum. Let $\omega_{2}$ be a cross-cut separating
$b^{(2)}$ from other marked points and infinity, and let $\gamma_{t}^{(2)}$
be the scaling limit of the interface starting from $b^{(2)}.$ Up
to hitting $\omega_{2}$, its law is given by Lemma \ref{lem: local_conv},
in particular, it is absolutely continuous with respect to SLE$_{16/3}.$
This means that $\gamma_{t}^{(2)}$ almost surely visits $\pa\Omega\setminus\left\{ b^{(2)}\right\} $,
and its hull almost surely covers a neighborhood of the boundary of
$b^{(2)}$, before hitting $\omega_{2}$. Therefore, the only way
$\gamma_{t}$ can visit $b^{(2)}$ is by connecting to $\gamma_{t}^{(2)}$
and traversing it in the reverse order, in which case it will hit
another boundary point first. The same argument applies to other marked
points.
\end{proof}
\begin{rem}
The conclusions of Lemma \ref{lem: tightness} imply that the map
$\gamma_{[0,\tau]}^{\H}\mapsto b_{[0,\tau]}^{(1)}$ is continuous
and injective on the support of the distribution of $\gamma_{[0,\tau]}^{\H}$.
Therefore, the law of $b_{[0,\tau]}^{(1)}$, specified by (\ref{eq: driving_process}),
determines uniquely the law of $\gamma_{[0,\tau]}^{\H}.$ A more subtle
technicality is whether the latter determines uniquely the law of
$\gamma_{[0,\tau]}$; the issue is that neither $\varphi$, nor $\varphi^{-1}$,
in general, induces a continuous injective map between the spaces
of curves. This question was answered in the positive in \cite{karrila2018limits}.
If the reader is ready to assume that $\pa\Omega$ is a curve, then
$\varphi^{-1}$ has a continuous extension to the boundary and thus
the map $\gamma_{[0,\tau]}^{\H}\mapsto\gamma_{[0,\tau]}$ is continuous
and injective, and the inverse map is injective, hence the issue disappears.
Note that by SLE duality, the boundary of SLE$_{16/3}$, and thus
of $\gamma_{[0,\tau]},$ is a. s. a curve; therefore, this regularity
assumption passes to the domains $\Omega_{L,R}$ defined below. 
\end{rem}

Lemma \ref{lem: hitting point} implies that for some random $2\leq J<2n$,
the marked points $b^{(2)},\dots,b^{(J)}$ belong to the boundary
of the same connected component $\Omega_{R}$ of $\Omega\setminus\gamma_{[0,\tau]}$,
and $b^{(J+1)},\dots,b^{(2n)}$ belong to the boundary of another
connected component $\Omega_{L}$. Similarly, for the discrete interface
$\gamma_{t}^{\delta}$, if $\tau^{\delta}$ is the first step after
which $b^{(2,\delta)},\dots,b^{(2n,\delta)}$ are not on the boundary
of the same connected component of $\Od_{t}$, there is some $J_{\delta}$
such that $b^{(2,\delta)},\dots,b^{(J_{\delta},\delta)}$ and $b^{(J_{\delta}+1,\delta)},\dots,b^{(2n)}$
are on the boundary of two different connected components $\Omega_{R}^{\delta}$
and $\Omega_{L}^{\delta}$ of $\Od_{\tau^{\delta}}$ respectively.
The domains $\Omega_{L,R}$, and similarly $\Od_{L,R}$, come naturally
equipped with the boundary conditions $\bcond_{L,R}$ that are inherited
from $\Omega$, with the additional change at $\gamma_{\tau}$ in
that of the two domains which contains an odd number of marked points.
We have the following convergence result:
\begin{lem}
\label{lem: stability_at_hitting}Under the coupling where $\gamma^{\delta}\to\gamma$
a. s., we can extract a subsequence $\delta_{k}$ such that, a. s.,
$J_{\delta_{k}}\to J$, $\Omega_{L,R}^{\delta_{k}}\stackrel{\text{Cara}}{\longrightarrow}\Omega_{L,R}$,
and $\bcond_{L,R}^{\delta_{k}}\to\bcond_{L,R}^{\delta_{k}}$, and
the latter approximation satisfies the boundary regularity assumption
(\ref{eq: bdry_reg}).
\end{lem}

\begin{proof}
Let $T_{1}^{\eps,\delta}$ be the first time the discrete interface
$\gamma_{t}^{\delta}$ comes to the distance $\eps^{2}$ from the
boundary arc $(b^{(2,\delta)},b^{(2n,\delta)}),$ and let $\omega^{\eps,\delta}$
denote the cross-cut in $\Od$ formed by the arc of the circle $|z-\gamma_{T_{1}^{\eps,\delta}}^{\delta}|=\eps$
which separates $\gamma_{T_{1}^{\eps,\delta}}^{\delta}$ from $b^{(1,\delta)}$.
Let $G^{\eps,\delta}$ be the event that all points of $\pa\Od$ separated
by $\omega^{\eps,\delta}$ from $b^{(1,\delta)}$ belong to the same
boundary arc $(b^{(i,\delta)},b^{(i+1,\delta)});$ as explained below,
$\lim\inf_{\delta\to0}\P(G^{\eps,\delta})\to1$ as $\eps\to0.$ 

Denote by $T_{2}^{\eps,\delta}$ be the first time after $T_{1}^{\eps,\delta}$
that $\gamma_{t}^{\delta}$ crosses $\omega^{\eps,\delta},$ and let
$E^{\eps,\delta}$ be the event that $\gamma_{[T_{1}^{\eps,\delta},T_{2}^{\eps,\delta}]}^{\delta}$
hits $\pa\Od$, and, moreover, $\mathrm{diam}(\gamma_{[T_{1}^{\eps,\delta},T_{2}^{\eps,\delta}]}^{\delta})\leq\sqrt{\eps}.$
We claim that there is a function $p(\eps)\to1$ as $\eps\to0$ such
that $\ind_{G^{\eps,\delta}}\cdot\P(E^{\eps,\delta}|\gamma_{[0,T^{\eps,\delta}]}^{\delta})\geq p(\eps)$
for all $\delta<\eps^{2}/10.$

\begin{figure}[h]
\includegraphics[width=0.4\textwidth]{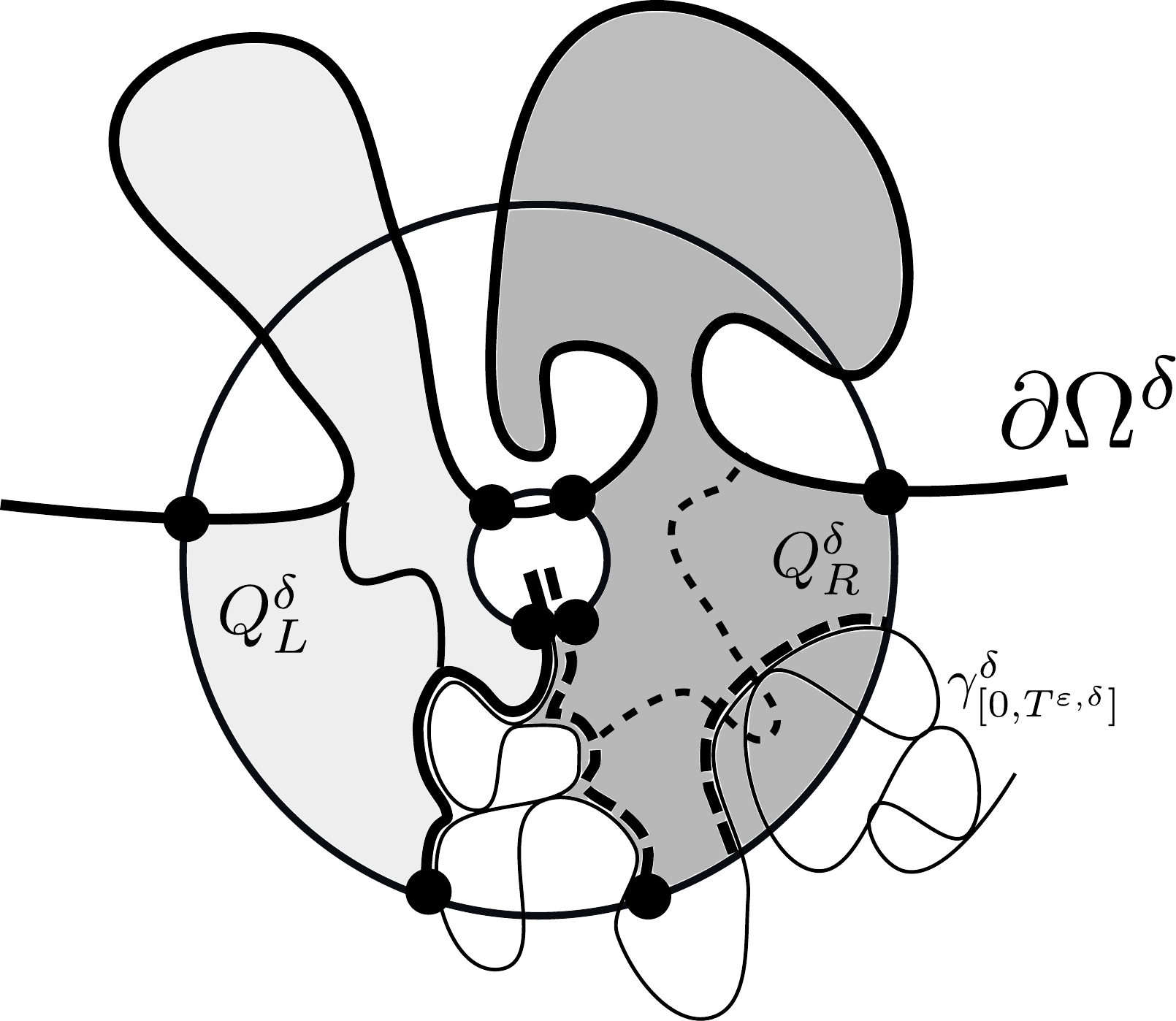}\caption{\label{fig: RSW}The moment $T^{\protect\eps,\delta}$ the curve $\gamma_{t}^{\delta}$
first comes close to the boundary of $\Omega^{\delta}$ away from
the arc $(b^{(2n,\delta)},b^{(1,\delta)}).$ The quads $Q_{L}^{\delta}$
and $Q_{R}^{\delta}$ are in light gray and dark gray respectively.
Their crossing, respectively, dual crossing, as shown in the picture,
forces the event $E^{\protect\eps,\delta}.$}

\end{figure}

Indeed, consider the annulus $A_{\eps}:=\{2\eps^{2}\leq|z-\gamma_{T^{\eps,\delta}}^{\delta}|\leq\eps\}.$
This annulus is crossed by $\pa\Od$ as well as by the interface $\gamma_{[0,T^{\eps,\delta}]}^{\delta},$
which means that it is also crossed by the wired $\partial_{\wired}$
and by the free $\partial_{\free}$ part of $\partial\Od_{T^{\eps,\delta}}\setminus\partial\Od$,
the left-hand side and the right-hand side of $\gamma_{[0,T^{\eps,\delta}]}^{\delta}$
respectively. We now consider the quads $Q_{R,L}^{\eps,\delta}$,
formed by a crossing of $A_{\eps}$ by $\pa_{\free}$ (respectively,
$\pa_{\wired}$) , two arcs of the circles $|z-\gamma_{T^{\eps,\delta}}^{\delta}|=\eps^{2}$
and $|z-\gamma_{T^{\eps,\delta}}^{\delta}|=2\eps^{3}$ up until their
first intersection with $\pa\Omega^{\delta}$, and a part of $\pa\Omega^{\delta}$
between these two intersection points, see Figure \ref{fig: RSW}.
Both quads have large modulus and thus, by RSW, if one puts wired
(respectively, free) boundary conditions on their boundaries, with
probability $q(\eps)\to1$ as $\eps\to0$, they are crossed by a dual
cluster (respectively, by a cluster) separating the arcs of the circles.
Conditionally on $\gamma_{[0,T^{\eps,\delta}]}^{\delta}$, such crossings
are even more likely by monotonicity, since any part of $\partial\Od_{T^{\eps,\delta}}$
that intersects the interior of the $Q_{R,L}^{\eps,\delta}$ is free
(respectively, wired). The coexistence of such crossings implies that
$\gamma_{[T_{1}^{\eps,\delta},T_{2}^{\eps,\delta}]}^{\delta}$ hits
$\pa\Od$, which therefore happens with probability at least $2q(\eps)-q(\eps)^{2}.$ 

For the diameter bound, consider the annulus $A'_{\eps}:=\{\eps\leq|z-\gamma_{T^{\eps,\delta}}^{\delta}|\leq\sqrt{\eps}/2\}.$
If the part of $\pa\Omega^{\delta}$ separated from $b^{(1,\delta)}$
by $\omega^{\eps,\delta}$ does not cross that annulus, then neither
does $\gamma_{[T_{1}^{\eps,\delta},T_{2}^{\eps,\delta}]}^{\delta}$
and we are done. Otherwise, consider the quad formed by arcs of the
inner and outer boundary of $A'_{\eps}$ and the first outward and
last inward crossing of it by $\pa\Od.$ By RSW and monotonicity,
with probability at least $q'(\eps)\to1$ as $\eps\to0$, this quad
is crossed by an open or by a dual-open path, which prevents $\gamma_{[T_{1}^{\eps,\delta},T_{2}^{\eps,\delta}]}^{\delta}$
from crossing $A'_{\eps},$ and thus from having diameter $\ge\sqrt{\eps}.$
All in all, we can take $p(\eps):=q'(\eps)+2q(\eps)-q(\eps)^{2}-1.$

With probability going to $1$ as $\eps\to0$, $\gamma_{[0,\tau]}$
stays at distance at least $3\eps$ from $[b^{(2)},b^{(J)}]\cup[b^{(J+1)},b^{(2n)}].$
On this event, for $\delta$ small enough, $\gamma_{[0,T^{\eps,\delta}]}^{\delta}$
stays at distance at least $2\eps$ from $[b^{(2,\delta)},b^{(J,\delta)}]\cup[b^{(J+1,\delta)},b^{(2n,\delta)}]$,
hence $G^{\delta,\eps}$ holds, and $E^{\eps,\delta}$ implies $J^{\delta}=J.$
Therefore, $J^{\delta}\to J$ in probability, and thus almost surely
along a subsequence. 

Any point $w$ of $\pa\Omega_{L}$ is either a point of $\pa\Omega$,
or a point of $\gamma$; in either case there is a sequence of points
$w^{\delta}\in\pa\Od\cup\gamma^{\delta}$ that approximates $w.$
We can find sequences $\eps_{k}$ and $\delta_{k}<\eps_{k}^{2}/10$
such that $\P(G^{\eps_{k},\delta_{k}})\geq1-2^{-k}$ and $p(\eps_{k})\geq1-2^{-k}$;
then, Borel--Cantelli and the above argument ensures that almost
surely, all but finitely many of $E^{\eps_{k},\delta_{k}}$ happen.
But if $w\in\Omega_{L}$, $\dist(w,\pa\Omega_{L})\geq2\eps$ and $B_{\eps}(w)\nsubseteq\Omega_{L}^{\delta_{k}},$
then, for $k$ large enough, this means that $E^{\eps_{k},\delta_{k}}$
has failed. Therefore, almost surely, for all $w\subset\Omega_{L}$
with rational coordinates, $B_{\dist(w;\pa\Omega_{L})/2}(w)\subset\Omega_{L}^{\delta_{k}}$
for $k$ large enough. This implies $\Omega_{L}^{\delta_{k}}\stackrel{\text{Cara}}{\longrightarrow}\Omega_{L};$
same argument applies to $\Omega_{R}^{\delta_{k}}.$ On $E^{\eps_{k},\delta_{k}},$the
cross-cut $\omega^{\eps_{k},\delta_{k}}$ actually separates the newly
created marked point in $\Od_{L}$ (if $J$ is even) or $\Od_{R}$
(if $J$ is odd) from the other marked points, and converges to $\gamma_{\tau}$,
which shows $\bcond_{L,R}^{\delta}\to\bcond_{L,R}.$ For the regularity
claim, which we also need to check only for the new marked point,
we can take $U^{\eps,\delta_{k}}$ to be the part of $\Od_{L}$ or
$\Od_{R}$ separated by the cross-cut $\omega^{\eps,\delta_{k}}$
from other marked points. The above argument shows that (\ref{eq: bdry_reg})
holds with $\eps'=\max(\sqrt{\eps},q'(\eps)).$ 
\end{proof}
Lemma \ref{lem: hitting point} allows for the following inductive
definition: 
\begin{defn}
\label{def: global_msle}The (global) multiple SLE$_{16/3}$ in $\Omega$
is a random collection of curves $\gamma^{(1)},\gamma^{(3)}\dots,\gamma^{(2n-1)}$
connecting $b^{(1)},\dots,b^{(2n-1)}\in\pa\Omega$ to $b^{(2\sigma(1))},\dots,b^{(2\sigma(n))}\in\pa\Omega$
respectively, for some random permutation $\sigma$ of $\{1,\dots,n\}$,
defined by the following properties: 
\end{defn}

\begin{itemize}
\item the marginal law of of the curve $\gamma_{[0,\tau]}^{(1)}$ started
from $b^{(1)}$ is given by the chordal Loewner evolution with the
driving process (\ref{eq: driving_process}). 
\item conditionally on $\gamma_{[0,\tau]}^{(1)}$, the law of ($\gamma^{(1)},\gamma^{(3)}\dots,\gamma^{(2n-1)}$)
is given by independent multiple SLE$_{16/3}$ processes in $\Omega_{L,R}$
in which the curve starting from $\gamma_{\tau}$ is concatenated
with $\gamma_{[0,\tau]}^{(1)}$. 
\item the base case $n=1$ is given by the chordal SLE$_{16/3}.$ 
\end{itemize}
\begin{proof}[Proof of Theorem \ref{thm: main} and Corollary \ref{cor: connection_prob}]
 Lemma \ref{lem: local_conv} and Lemma \ref{lem: stability_at_hitting}
ensure that the marginal distribution of $\gamma_{[0,\tau^{\delta}]}^{(1,\delta)}$
converges to that of $\gamma_{[0,\tau]}^{(1)}$. The full interface
$\gamma_{[0,N]}^{(1,\delta)}$ consist of $\gamma_{[0,\tau^{\delta}]}^{(1,\delta)}$,
the part from $\tau^{\delta}$ until $\gamma_{t}^{(1,\delta)}$ re-enters
$\Od_{L}$ (if $J$ is odd) or $\Od_{R}$ (if J is even), and then
the part of the interface in $\Od_{L}$ or $\Od_{R}$. The proof of
Lemma \ref{lem: stability_at_hitting} ensures that the diameter of
the middle part goes to zero in probability. Thus, the domain Markov
property, Lemma \ref{lem: stability_at_hitting} and the induction
hypothesis imply that the conditional distribution of $(\gamma^{(1,\delta)}$,$\dots$,$\gamma^{(2n-1,\delta)})$
given $\gamma_{[0,\tau^{\delta}]}^{(1,\delta)}$ converges to the
conditional distribution of $(\gamma^{(1)}$,$\dots$,$\gamma^{(2n-1)})$
given $\gamma_{[0,\tau]}^{(1)}$, at least along a subsequence $\delta_{k}.$
Hence the full law of $(\gamma^{(1,\delta)},\dots,\gamma^{(2n-1,\delta)})$
converges to that of $(\gamma^{(1)},\gamma^{(3)}\dots,\gamma^{(2n-1)})$
along $\delta_{k}$. Since such a subsequence can be extracted from
any sequence of $\delta\to0$, no extraction is in fact needed. The
random variable $\ind_{\pi^{\delta}=\pi}$ is a continuous function
of the collection if interfaces and the latter has a conformally invariant
scaling limit, therefore, Corollary \ref{cor: connection_prob} also
follows.
\end{proof}
\bibliographystyle{plain}
\bibliography{Mixedcorr}

\end{document}